\documentclass{article}

\PassOptionsToPackage{numbers,compress}{natbib}
\usepackage[main,final]{neurips_2026}

\usepackage[T1]{fontenc}
\usepackage[utf8]{inputenc}
\usepackage{microtype}
\usepackage{amsmath,amssymb,amsthm,mathtools}
\usepackage{aliascnt}
\usepackage{array}
\usepackage{booktabs}
\usepackage{enumitem}
\usepackage{tabularx}
\usepackage{xcolor}
\usepackage{graphicx}
\usepackage{url}
\usepackage[colorlinks=true,linkcolor=blue!50!black,citecolor=blue!50!black,urlcolor=blue!50!black]{hyperref}
\hypersetup{pdftitle={When One Leak Pays Forever: Context Binding and the Price of Deterring Collusion},
  pdfauthor={Tingyi Lin, Shawn Yu, Ruoran Lai, Huanxi Zhang}}
\usepackage{orcidlink}
\usepackage[nameinlink,noabbrev]{cleveref}

\numberwithin{equation}{section}
\newtheorem{theorem}{Theorem}[section]
\crefname{theorem}{Theorem}{Theorems}
\Crefname{theorem}{Theorem}{Theorems}
\newaliascnt{lemma}{theorem}
\newtheorem{lemma}[lemma]{Lemma}
\aliascntresetthe{lemma}
\crefname{lemma}{Lemma}{Lemmas}
\Crefname{lemma}{Lemma}{Lemmas}
\newaliascnt{proposition}{theorem}
\newtheorem{proposition}[proposition]{Proposition}
\aliascntresetthe{proposition}
\crefname{proposition}{Proposition}{Propositions}
\Crefname{proposition}{Proposition}{Propositions}
\newaliascnt{corollary}{theorem}
\newtheorem{corollary}[corollary]{Corollary}
\aliascntresetthe{corollary}
\crefname{corollary}{Corollary}{Corollaries}
\Crefname{corollary}{Corollary}{Corollaries}
\theoremstyle{definition}
\newaliascnt{definition}{theorem}
\newtheorem{definition}[definition]{Definition}
\aliascntresetthe{definition}
\crefname{definition}{Definition}{Definitions}
\Crefname{definition}{Definition}{Definitions}
\newaliascnt{example}{theorem}
\newtheorem{example}[example]{Example}
\aliascntresetthe{example}
\crefname{example}{Example}{Examples}
\Crefname{example}{Example}{Examples}
\theoremstyle{remark}
\newaliascnt{remark}{theorem}
\newtheorem{remark}[remark]{Remark}
\aliascntresetthe{remark}
\crefname{remark}{Remark}{Remarks}
\Crefname{remark}{Remark}{Remarks}
\theoremstyle{plain}
\newtheorem{maintheorem}{Theorem}

\newcommand{\Players}{[n]}
\newcommand{\A}{\mathcal{A}}
\newcommand{\Rset}{\mathcal{R}}
\newcommand{\LambdaS}{\Lambda}
\newcommand{\N}{\mathbb{N}}

\title{When One Leak Pays Forever: Context Binding and the Price of Deterring Collusion}

\author{%
  Tingyi Lin$^{1\dagger}$\,\orcidlink{0009-0001-9576-0912}, Shawn Yu$^{2}$, Ruoran Lai$^{3}$, Huanxi Zhang$^{4}$\\
  $^{1}$Tropic Research Institute\quad $^{2}$Boston College\\
  $^{3}$Sun Yat-sen University\quad $^{4}$University of Wisconsin--Madison
}
\date{}

\begin{document}
\maketitle
{\renewcommand{\thefootnote}{\fnsymbol{footnote}}%
\footnotetext[2]{Correspondence to: tingyi3@illinois.edu}}

\begin{abstract}
A coalition that deviates once can profit many times when what it sells keeps
working. In a threshold-encrypted mempool, a leading defense against maximal
extractable value (MEV), a quorum of the decryption committee that sells its
decryption capability to a front-runner exposes every later block that the
capability still decrypts. We ask how large a penalty, such as
slashable stake, deters this kind of collusion. In our repeated game, a single
leak by any coalition in a monotone family of authorized coalitions (for
example, any $k$ of the $n$ committee members) unlocks a set of future rounds,
costs a one-time penalty, and ends the coalition's participation. We show that
every dynamic deviation reduces to choosing a leak time, so deterrence holds if
and only if
each coalition's penalty covers the largest discounted value that a single leak
reaches. Without discounting, over $T$ rounds of unit value full reuse needs a
penalty of $T$ while binding each leak to its own round needs $1$, so no penalty
that is constant in the horizon deters unbounded reuse; a reuse window of $w$
rounds costs at most $w$ times the largest per-round value. The cheapest profile of
per-party stakes that deters every coalition solves a covering linear program.
For blockchain design, per-epoch keys cut the required stake from the value of a
key's lifetime to the value of one epoch; we calibrate the gap on Ethereum
front-running data and place Ferveo and Shutter in the model. The analysis extends to sealed-bid auctions,
multi-authority voting, and federated learning under a shared key.
\end{abstract}

\section{Introduction}

\begin{flushright}\small
\textit{Three may keep a secret, if two of them are dead.}\\
--- Benjamin Franklin, \textit{Poor Richard's Almanack}, 1735
\end{flushright}

Deterrence is usually priced one act at a time: a deviation is deterred when
the penalty it risks exceeds what it gains \cite{Becker68}. In a repeated game
the gain is a continuation value \cite{MailathSamuelson06}, and the comparison
stays simple as long as each deviation is a single act. It stops being simple
when a coalition deviates once by selling a capability that keeps working. The
buyer then collects in later rounds as well, the coalition's gain becomes a
stream, and the penalty needed to deter the sale depends on how many future
rounds that one sale reaches.

Threshold-encrypted mempools are the leading example
\cite{BebelOjha22,DziembowskiFaustLuhn24,KavousiLeJovanovicDanezis23}. Users
encrypt transactions to a committee key, and a batch is decrypted only after its
order is fixed, which withholds the information that front-running and
sandwich attacks exploit
\cite{DaianGoldfederKellLiZhaoBentovBreidenbachJuels19,TorresCaminoState21}. Any
authorized quorum of the committee can still reconstruct a decryption
capability early and sell it to a front-runner. When the committee key lives
for many epochs, one sale reveals every later batch encrypted under it.
Context-dependent threshold decryption \cite{BonehBunzNayakRotemShoup25} binds
each decryption share to a public context such as an epoch number. When every
epoch has its own key, or when the long-term key shares sit in modules that
decrypt only the current context's ciphertexts and are replaced as soon as a
leak becomes public, a sale in one epoch reveals nothing about the next. Sealed-bid auctions,
multi-authority voting, and federated learning under a collective key share
this pattern of one leak and many rounds (\Cref{sec:apps}).

We study a repeated game among $n$ parties who jointly hold a decryption
capability. A monotone family $\A$ of \emph{authorized} coalitions, such as all
coalitions of at least $k$ parties, lists who can leak. The \emph{stage value}
$v_t(S)$ is what coalition $S$ can earn by selling early access to the hidden
data of round $t$. A leak is public and ends the leaking parties' role in the
protocol, so each coalition leaks at most once. The \emph{reachability set}
$\Rset_S(t)$ is the set of rounds whose data a buyer can read after $S$ leaks at
round $t$; a long-lived key makes it every later round, and context binding,
which ties everything a coalition can leak to its round, makes it round $t$
alone (\Cref{fig:reach}). The \emph{exposure} $\kappa(S)$ is the one-time
penalty that $S$ bears for its leak, such as slashed stake, a tracing liability,
or forfeited rewards. Values are discounted by a factor $\delta$ per round and
are transferable within a coalition, and the discounted value that a single leak
reaches is bounded across leak times. A coalition is \emph{deterred} if, after
any history without a leak, it gains nothing from any deviation strategy, and
honest execution is \emph{coalition-safe} if every coalition is deterred.

\begin{figure}[t]
\centering
\includegraphics{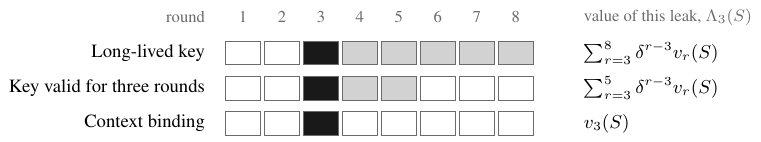}
\caption{A leak at round $3$ under three keying designs. Dark: the round of
the leak; gray: later rounds whose data the buyer can still read; white: rounds
that stay hidden. The right column is the value $\LambdaS_3(S)$ of this leak, and
\Cref{thm:master} requires $\kappa(S)\geq\sup_t\LambdaS_t(S)$.}
\label{fig:reach}
\end{figure}

Our first result prices deterrence exactly.

\begin{maintheorem}[Exact deterrence threshold]
A coalition $S$ is deterred if and only if its exposure $\kappa(S)$ is at least
the discounted value $\LambdaS_t(S)$ of the rounds that a leak by $S$ at round
$t$ reaches, for every $t$:
\[
\kappa(S)\geq\sup_{t}\LambdaS_t(S),
\qquad
\LambdaS_t(S)=\sum_{r\in\Rset_S(t)}\delta^{r-t}v_r(S).
\]
Honest execution is coalition-safe if and only if this holds for every
authorized coalition.
\end{maintheorem}

Theorem~A is \Cref{thm:master}. A coalition may time its leak to any history of
play, and one leak may unlock a nonconsecutive set of later rounds. Because the
leak is the coalition's last act, every deviation is worth exactly what its leak
is worth (\Cref{lem:normal}), so the repeated interaction enters only through
the reachability sets. The one-shot rule of Becker~\cite{Becker68} is the case
$\Rset_S(t)=\{t\}$.

Our second result measures what reuse costs.

\begin{maintheorem}[Reuse sets the price]
Fix an authorized coalition $S$.
\begin{enumerate}[label=(\roman*),leftmargin=2em]
\item On a $T$-round instance without discounting in which every stage value of
$S$ equals $1$, deterring $S$ requires exposure $T$ when a leaked capability
stays valid in every later round, and exposure $1$ when it is bound to its own
round.
\item Without discounting and with unit stage values, $S$ is deterred exactly
when its exposure is at least $\rho(S)$, the largest number of rounds that one
leak by $S$ reaches. No exposure that is constant in the horizon therefore
deters $S$ along a family of instances in which $\rho(S)$ grows without bound.
\item If every leak by $S$ reaches at most $\rho(S)$ rounds and every stage
value of $S$ is at most $\bar v(S)$, then exposure
$\bar v(S)\sum_{j=0}^{\rho(S)-1}\delta^{j}$, which is at most
$\rho(S)\,\bar v(S)$, deters $S$. No smaller exposure does when a leak at some
round $t$ reaches exactly the rounds $t,\dots,t+\rho(S)-1$, each of value
$\bar v(S)$.
\end{enumerate}
\end{maintheorem}

Our third result splits the requirement across parties.

\begin{maintheorem}[Cheapest deterring stakes]
Suppose party $i$ posts exposure $d_i\geq 0$, for example bonded stake, and a
coalition's exposure is the sum over its members. A profile $(d_i)$ deters
every authorized coalition exactly when it satisfies one covering constraint per
authorized coalition, so the cheapest such profile solves a covering linear
program. If the authorized coalitions are those with at least $k$ members, for
some $1\leq k\leq n$, all parties post the same amount, and the supremum of a
coalition's reachable value depends only on its size $m$, say $\lambda_m$,
then the smallest deterring amount per party is $\max_{k\leq m\leq n}\lambda_m/m$.
\end{maintheorem}

Theorem~B collects \Cref{thm:separation}, \Cref{prop:reuse-mult}, and
\Cref{cor:bounded-reuse}; Theorem~C collects \Cref{cor:hetero-lp,cor:uniform}.
\Cref{sec:beyond} extends Theorem~A to probabilistic and time-distributed
exposure, where liability that grows with use restores deterrence under long
reuse.

\paragraph{Techniques.}
Three objects describe a leak: who can produce it ($\A$), what each round is
worth ($v_t$), and how far one leak propagates ($\Rset_S$). Only the third
involves time. \Cref{lem:normal} shows that it carries all of the dynamics, and
exact deterrence becomes a supremum over first-leak times. Two smaller state
variables fail. Comparing $\kappa(S)$ with the largest stage value, as a
one-shot model does, is exact under context binding and can be off by a factor
of $T$ under full reuse (Theorem~B(i)). Counting reachable rounds is exact with
unit values and no discounting (\Cref{prop:reuse-mult}); otherwise the identity
and timing of the reached rounds matter. Since the reachability set is the only dynamic object,
context binding acquires a one-line game-theoretic meaning,
$\Rset_S(t)=\{t\}$, and Theorem~B prices it.

\section{Related Work}
\label{sec:related}

\paragraph{Deterrence and coalition deviations.}
Comparing an expected penalty with the gain from an offense goes back to
Becker~\cite{Becker68}, and \Cref{cor:oneshot} is that rule for a single leak.
Deviations by groups are the subject of strong equilibrium \cite{Aumann59} and
coalition-proof equilibrium \cite{BernheimPelegWhinston87}. Our benchmark is
lighter: a coalition deviates from honest execution and, since utility is
transferable inside it, acts as one player. In repeated games the one-shot
deviation principle reduces deviations to single-period ones when payoffs are
continuous at infinity \cite{MailathSamuelson06}. \Cref{lem:normal} plays that
role here, with one difference: a leak is irreversible, so the deviation that
matters is the first leak, and the reachability set fixes its value. The new
element is this state variable: it suffices for exact deterrence over any
monotone family of coalitions, and it separates reusable from context-bound
deviations by a factor of $T$.

\paragraph{Stake, slashing, and economic security.}
Proof-of-stake protocols deter misbehavior by slashing bonded stake
\cite{ButerinGriffith17}, and Budish, Lewis-Pye, and Roughgarden~\cite{BudishLewisPyeRoughgarden24} study when
attacks on permissionless consensus can be made expensive. In those analyses the
slashed stake is weighed against the profit of an attack. Our exposure $\kappa$
plays the role of the slashed stake, and the profit it must outweigh is a stream
over every round that one leak unlocks.

\paragraph{Incentives in collaborative learning.}
Mechanisms for collaborative and federated learning deter individual
participants who under-collect data, fabricate it, or send dishonest updates
\cite{DornerKonstantinovPashalievVechev23,ChenZhuKandasamy23}. The deviation
studied here is collective, a quorum that can open data encrypted under a shared
key, and its payoff accrues over every round the key protects.

\paragraph{Rational cryptography.}
Rational secret sharing \cite{HalpernTeague04,AbrahamDolevGonenHalpern06} and
rational cryptography more broadly \cite{MiltersenNielsenTriandopoulos09} ask
when cryptographic protocols implement desired equilibria among self-interested
parties. There the parties reconstruct or compute a secret among themselves.
Here an authorized coalition sells a functional decoder to an outsider, and the
difficulty is the continuation value of a decoder reused across contexts, which
one-shot reconstruction never faces.

\paragraph{Threshold decryption, accountability, and context binding.}
Ferveo \cite{BebelOjha22}, Shutter \cite{DziembowskiFaustLuhn24}, and BlindPerm
\cite{KavousiLeJovanovicDanezis23} hide transaction content until its order is
fixed, and ETHTID \cite{StengeleRaiberMuellerQuadeHartenstein21} combines
scheduled disclosure with stake-backed penalties. Threshold traitor tracing
\cite{BonehPartapRotem23} traces a leaked decoder to a member of the quorum that
produced it, and context-dependent threshold decryption
\cite{BonehBunzNayakRotemShoup25} ensures that decryption shares issued under
different contexts cannot be combined, a guarantee against fewer than a
threshold of corrupt parties. These works build and analyze the cryptography.
In our model their accountability mechanisms determine the exposure $\kappa$,
and their keying choices, together with what a quorum can hand to a buyer,
determine the reachability sets (\Cref{sec:apps-mempool}).

\paragraph{MEV and order fairness.}
Transaction reordering and front-running were identified by
Daian et~al.~\cite{DaianGoldfederKellLiZhaoBentovBreidenbachJuels19} and measured at scale
by Torres et~al.~\cite{TorresCaminoState21}; Kulkarni et~al.~\cite{KulkarniDiamandisChitra22} study their
game-theoretic structure in constant-function market makers. Order fairness,
inclusion lists, and protected order flow are analyzed in
\cite{LiHuWangDuanWang23,WadhwaMaThieryMonnotZanoliniZhangNayak25,BabelJeanLouisJiMisraKelkarMudiyanselageMillerJuels24},
the economics of mempool privacy in \cite{RondeletKilbourn23}, and mitigation
more broadly in the surveys
\cite{AlipanahlooHafidZhang24,MaterwalaNaikTahaAbedSvetinovic24}. Accountable
mempools such as L{\O} \cite{NasrulinIshmaevDecouchantPouwelse23} act on the
exposure $\kappa$ or on its time-distributed version. We take the stakes of
these applications as the stage values $v_t(S)$ and study the leak that
precedes honest release.

\section{Model and Preliminaries}
\label{sec:model}

\subsection{Access Structures and Repeated Contexts}

We fix a set of parties $\Players = \{1,\dots,n\}$.

\begin{definition}[Monotone access structure]
A family $\A \subseteq 2^{\Players}\setminus\{\varnothing\}$ is a \emph{monotone access structure} if
$S \in \A$ and $S \subseteq S'$ imply $S' \in \A$. A coalition $S$ is
\emph{authorized} if $S \in \A$ and \emph{unauthorized} otherwise.
\end{definition}

In the language of cooperative game theory, $\A$ is the family of winning
coalitions of a monotone simple game \cite{Shapley62,TaylorZwicker99}. A
$k$-out-of-$n$ policy is the weighted voting game with unit weights and quota
$k$, and a committee whose shares are weighted by stake is a general weighted
voting game.

The parties run a service over rounds $t=1,2,\dots,T$, with horizon
$T \in \N \cup \{\infty\}$ and discount factor $\delta\in(0,1]$ per round.
Round $t$ has a public context tag $\chi_t$, whatever a leaked capability can be
bound to (an epoch or batch in an encrypted mempool, an auction instance, a
training round), and hidden data that are due to become public at that round's
honest release point.

\subsection{Leakage Values and One-Time Accountability Exposure}

For each coalition $S \subseteq \Players$ and round $t$, let
$v_t(S) \in \mathbb{R}_{\ge 0}$ denote the coalition's \emph{stage
value}: its transferable-utility gain from an outside buyer's unauthorized early
access to round $t$'s hidden data, whether a leak by $S$ at round $t$ or an
earlier one provides that access. It is a summary statistic that may absorb side
payments from a buyer, front-running rents, or any other value of early access,
and the theorems need no decomposition of it.

Accountability is captured by a public \emph{one-time exposure function}
$\kappa: 2^{\Players} \to \mathbb{R}_{\ge 0}$: a coalition $S$ that leaks bears
the loss $\kappa(S)$, measured in present value at the round of the leak. This
loss may represent slashable stake, legal penalties, lost future rewards,
reputational damage, or an effective expected penalty after probabilistic
tracing.

\begin{remark}
The exposure is an effective quantity. If tracing identifies at least one member
of $S$ with probability $p(S)$ and the present-value fine after a successful
trace is $f(S)$, the risk-neutral exposure is $\kappa(S)=p(S)f(S)$, the
coalition's total present-value liability from the leak.
\end{remark}

\subsection{Reachability Semantics}
\label{sec:reach}

The central modeling object is the set of future contexts compromised by one
leak. A \emph{successful leak}, the sale of a functional decoder to a buyer, is
public, and we assume that it is terminal: the leaking parties take no further
part in the protocol, as when they are ejected from the committee. Each
coalition therefore leaks at most once.

\begin{definition}[Leakage reachability]
For each authorized coalition $S$ and round $t$, let $\Rset_S(t)$, a subset of
$\{t,t+1,\dots,T\}$ that contains $t$, be the set of rounds whose hidden data
become accessible to $S$'s buyer as a result of a successful leak by $S$ in
round $t$. For unauthorized coalitions we set $\Rset_S(t)=\varnothing$.
\end{definition}

Since the leak is terminal, $\Rset_S(t)$ records everything it yields: every
later round when the leaked object stays valid forever, and round $t$ alone
when nothing the coalition can leak opens another context.

\begin{definition}[Discounted leakage value]
\label{def:lambda}
For every coalition $S$ and round $t$, define
\[
\LambdaS_t(S) = \sum_{r \in \Rset_S(t)} \delta^{r-t} v_r(S)
\in[0,\infty].
\]
Throughout, we assume that $\sup_{t\leq T}\LambdaS_t(S)<\infty$ for every
$S\in\A$. The assumption holds automatically when $T<\infty$, and when
$T=\infty$ with $\delta<1$ and bounded stage values; when $T=\infty$ and
$\delta=1$, it requires the undiscounted sums $\sum_{r\in\Rset_S(t)}v_r(S)$ to
be bounded in $t$.
\end{definition}

Thus $\LambdaS_t(S)$ is the present-value benefit, at round $t$, of a leak at
round $t$. When the supremum is infinite, no finite one-time exposure deters
$S$.

\begin{definition}[Reusable and context-isolating schemes]
We say that a deployment is:
\begin{enumerate}[label=(\roman*),leftmargin=2em]
\item \emph{fully reusable} if $\Rset_S(t)=\{t,t+1,\dots,T\}$ for every
authorized coalition $S$ and every round $t$;
\item \emph{context-isolating} if $\Rset_S(t)=\{t\}$ for every authorized
coalition $S$ and every round $t$.
\end{enumerate}
\end{definition}

\begin{remark}
Context isolation, which we also call context binding, says that nothing an
authorized coalition can leak in context $\chi_t$ is useful in any context
$\chi_{t'}\neq\chi_t$. Context-dependent threshold decryption
\cite{BonehBunzNayakRotemShoup25} ensures that decryption shares issued under
different contexts cannot be combined, which protects a ciphertext against fewer
than a threshold of corrupt parties. An authorized quorum, however, also holds
long-lived key shares that decrypt under every context. The scheme is
context-isolating when these shares sit in \emph{revocable modules}, which
decrypt only ciphertexts bound at encryption to the current context and which the
protocol revokes and replaces before the next context once a leak becomes
public, whoever supplied it; shares that merely cannot be exported do not
suffice, because the modules themselves can be handed to a buyer.
\end{remark}

\subsection{Deviation Benchmark}

The object of interest is whether honest
execution is safe against coalition deviations under transferable utility.

\begin{definition}[Deterrence and coalition safety]
Fix $(\A,\{v_t\},\kappa,\{\Rset_S(t)\},\delta)$. A coalition
$S\subseteq\Players$ is \emph{deterred} if, in every round $t$ and after every
public history in which no successful leak has occurred before round $t$, no
deviation strategy gives $S$ a positive gain in transferable utility, in present
value at round $t$, over continuing to behave honestly forever. Honest execution
is \emph{coalition-safe} if every coalition is deterred.
\end{definition}

Two features of this benchmark depart from a plain repeated game. Utility is
transferable within a coalition, so a coalition deviates as a single player: a
deviation by $S$ is a leak made jointly by all members of $S$, who share its
value and bear $\kappa(S)$, and a leak supplied by a proper subset
$S_1\subsetneq S$ is a deviation of $S_1$. This attribution is a modeling
choice: if tracing exposes only the members whose shares were used, a coalition
could leak through a minimal authorized subset and share the gain, so
deterrence would have to hold against that subset's exposure.
Exposure is charged once, at the leak, and actions other than leaking carry no
payoff; \Cref{sec:beyond} replaces the one-time
charge by probabilistic and time-distributed exposure. The games played after
plaintexts become public are outside the model.

\section{Main Results}
\label{sec:main}

The first theorem gives the exact deterrence threshold in the general
reachability model. Appendix~\ref{sec:proofs} proves the statements of this section
in order.

\begin{theorem}[Exact characterization]
\label{thm:master}
A coalition $S$ is deterred if and only if
\[
\kappa(S) \geq \sup_{t \leq T}\LambdaS_t(S).
\]
Every unauthorized coalition meets this condition, so honest execution is
coalition-safe if and only if it holds for every $S\in\A$.
\end{theorem}

\paragraph{Proof idea.}
We fix a coalition $S$ and a deviation strategy, and we let $\tau$ be the round
of its leak. Because the leak is terminal, the buyer reads exactly the rounds in
$\Rset_S(\tau)$, and the exposure is charged once, at $\tau$; so the deviation
is worth $\LambdaS_\tau(S)-\kappa(S)$ in present value at $\tau$
(\Cref{lem:normal}). The condition of \Cref{thm:master} makes every such value
nonpositive, which gives sufficiency. For necessity, a coalition with
$\LambdaS_t(S)>\kappa(S)$ plays honestly until round $t$ and leaks there. The
dynamics enter only through the choice of $\tau$, and the reachability set
records every consequence of the leak, so the theorem needs no model of the
market that follows it.

The remaining results answer four design questions: what one-shot reasoning
misses (\Cref{cor:oneshot,cor:reusable,cor:context}), how large the gap can be
(\Cref{thm:separation,prop:reuse-mult}), what bounded reuse buys
(\Cref{cor:bounded-reuse}), and how to split the requirement across parties
(\Cref{cor:uniform,cor:hetero-lp}). The one-shot rule is exact with a single
round and under context binding, while a capability that never expires makes
the exposure cover the whole remaining stream.

\begin{corollary}[One-shot threshold]
\label{cor:oneshot}
In a one-round environment, honest execution is coalition-safe if and only if
$\kappa(S)\geq v_1(S)$ for every $S\in\A$.
\end{corollary}

\begin{corollary}[Fully reusable leakage]
\label{cor:reusable}
Suppose the deployment is fully reusable. Then honest execution is
coalition-safe if and only if
\[
\kappa(S)
\geq
\sup_{t \leq T}\sum_{r=t}^{T}\delta^{r-t} v_r(S)
\qquad\text{for every } S \in \A.
\]
In particular, if $T=\infty$, $\delta<1$, and $v_r(S)=v(S)$ for all $r$, then
the exact threshold is ${\kappa(S)\geq v(S)/(1-\delta)}$.
\end{corollary}

This requirement grows without bound as $\delta\to 1$.

\begin{corollary}[Context isolation collapses continuation value]
\label{cor:context}
In a context-isolating deployment, honest execution is coalition-safe if and
only if $\kappa(S)\geq\sup_{t\leq T}v_t(S)$ for every $S\in\A$.
\end{corollary}

The gap between these two designs can be as large as the horizon.

\begin{theorem}[Sharp reusable-versus-context-isolating separation]
\label{thm:separation}
Let $T\in\N$, and let $S^\star$ be an authorized coalition with
$v_t(S^\star)=1$ for every $t\leq T$ in an instance without discounting
($\delta=1$). Then:
\begin{enumerate}[label=(\roman*),leftmargin=2em]
\item under fully reusable leakage, $S^\star$ is deterred if and only if
$\kappa(S^\star)\geq T$;
\item under context isolation, $S^\star$ is deterred if and only if
$\kappa(S^\star)\geq 1$.
\end{enumerate}
Moreover, $\sup_{t\leq T}\sum_{r=t}^{T}\delta^{r-t}v_r(S)\leq T\sup_{t\leq T}
v_t(S)$ for every $T$-round instance and every coalition $S$, so no instance
separates the two thresholds by more than a factor of $T$.
\end{theorem}

A committee that budgets one round of value under a long-lived key is short by a
factor of $T$, because one leak collects every round; in every design the gap
is driven by how many rounds one leak can reach.

\begin{proposition}[Reuse multiplicity threshold]
\label{prop:reuse-mult}
Fix an authorized coalition $S$, let $\rho(S)=\sup_{t\leq T}|\Rset_S(t)|$, and
suppose $\delta=1$ and $v_r(S)=1$ for every round $r$. Then $S$ is deterred if
and only if $\kappa(S)\geq\rho(S)$. Consequently, along a family of deployments
in which $\rho(S)$ grows without bound, no exposure that is constant in the
horizon deters $S$.
\end{proposition}

A designer who cannot bound $\rho(S)$ cannot certify deterrence with a fixed
exposure, while a bound on it gives the closed form below; there $\rho(S)$
denotes any such bound, which may exceed the supremum of \Cref{prop:reuse-mult}.

\begin{corollary}[Bounded reuse multiplicity under discounting]
\label{cor:bounded-reuse}
Fix an authorized coalition $S$ and suppose that
$|\Rset_S(t)| \leq \rho(S)$ for every round $t$. Assume also that
$v_r(S)\leq \bar v(S)$ for every round $r$. Then $S$ is deterred whenever
\[
\kappa(S)
\geq
\bar v(S)\sum_{j=0}^{\rho(S)-1}\delta^j
=
\begin{cases}
\bar v(S)\dfrac{1-\delta^{\rho(S)}}{1-\delta}, & \delta<1,\\[0.8em]
\rho(S)\bar v(S), & \delta=1.
\end{cases}
\]
Moreover, this bound is tight whenever some round $t$ satisfies
$\Rset_S(t)=\{t,t+1,\dots,t+\rho(S)-1\}$ and
$v_{t+j}(S)=\bar v(S)$ for every $j=0,\dots,\rho(S)-1$.
\end{corollary}

Because key rotation every $\rho(S)$ rounds keeps every reachability set within
$\rho(S)$ rounds when no leak can include a later key, it caps the multiplier at
$\rho(S)$ for every $\delta\leq 1$. For the common equal-stake threshold committee, the
general theorem becomes an exact formula.

\begin{corollary}[Uniform exposure for $k$-out-of-$n$ committees]
\label{cor:uniform}
Suppose $\A=\{S \subseteq \Players : |S|\geq k\}$ for some $1\leq k\leq n$ and
$\kappa(S)=d|S|$ for a common per-party exposure $d \geq 0$. Then the
minimal such $d$ that makes honest execution coalition-safe is
\[
d^\star
=
\max_{S \in \A}
\frac{\sup_{t \leq T}\LambdaS_t(S)}{|S|}.
\]
If moreover $\sup_t \LambdaS_t(S)$ depends only on $|S|$, say
$\sup_t \LambdaS_t(S)=\lambda_{|S|}$, then
\[
d^\star
=
\max_{k\leq m\leq n}\frac{\lambda_m}{m}.
\]
\end{corollary}

The binding coalition is the size at which reachable value per member peaks,
not necessarily the smallest; with unequal amounts the cheapest profile solves a
covering program.

\begin{corollary}[Heterogeneous additive exposure as a covering LP]
\label{cor:hetero-lp}
Suppose the deployment chooses nonnegative per-party exposures
$d_1,\dots,d_n$ and coalition exposure is additive:
\[
\kappa(S)=\sum_{i \in S} d_i.
\]
Then honest execution is coalition-safe if and only if
$\sum_{i\in S}d_i\geq\sup_{t\leq T}\LambdaS_t(S)$ for every $S\in\A$.
Consequently, the minimum-total-exposure design problem is the linear program
\[
\begin{aligned}
\min\quad & \sum_{i=1}^{n} d_i\\
\mathrm{s.t.}\quad
& \sum_{i \in S} d_i \geq \sup_{t \leq T}\LambdaS_t(S)
&& \forall S \in \A,\\
& d_i \geq 0
&& \forall i \in \Players.
\end{aligned}
\]
\end{corollary}

\Cref{thm:master} needs neither additive nor symmetric exposure nor a threshold
family $\A$, so it covers heterogeneous slashing and weighted or hierarchical
access policies; \Cref{tab:summary} collects the thresholds.

\begin{table}[t]
\centering
\small
\begin{tabularx}{\linewidth}{@{}>{\raggedright\arraybackslash}p{0.42\linewidth}>{\raggedright\arraybackslash}p{0.12\linewidth}>{\raggedright\arraybackslash}X@{}}
\toprule
Setting & $\Rset_S(t)$ & Threshold for coalition $S$ \\
\midrule
Context-isolating & $\{t\}$ & $\sup_t v_t(S)$ \\
Fully reusable & $\{t,\dots,T\}$ &
$\sup_t \sum_{r=t}^T \delta^{r-t} v_r(S)$ \\
Bounded reuse $|\Rset_S(t)|\leq \rho(S)$, $v_r(S)\leq \bar v(S)$ &
arbitrary & at most $\bar v(S)\sum_{j=0}^{\rho(S)-1}\delta^j$ \\
Reuse multiplicity $\rho(S)$, unit values, $\delta=1$ & arbitrary & $\rho(S)$ \\
Equal exposure $d|S|$, $k$-out-of-$n$ & arbitrary & $d^\star = \max_{S \in \A}\sup_t\LambdaS_t(S)/|S|$ \\
Additive exposure $\sum_{i \in S} d_i$ & arbitrary & choose $d_i$ by the LP of \Cref{cor:hetero-lp} \\
\bottomrule
\end{tabularx}
\caption{Thresholds by reachability pattern: exact where available, sharp bounds
otherwise.}
\label{tab:summary}
\end{table}

\section{Beyond a One-Time Exposure}
\label{sec:beyond}

Two relaxations keep the first-leak argument and change only the cost side of
the comparison; both assume risk-neutral coalitions and detection probabilities
and fines that the deployment fixes in advance.

\begin{proposition}[Probabilistic exposure]
\label{prop:prob}
Suppose coalition $S$ is exposed after a successful leak with probability
$p(S)\in[0,1]$, and conditional on exposure pays a deterministic fine
$f(S)\geq 0$. If the coalition is risk neutral, then it is deterred if and only
if $p(S)\,f(S)\geq\sup_{t\leq T}\LambdaS_t(S)$.
\end{proposition}

Tracing mechanisms such as threshold traitor tracing \cite{BonehPartapRotem23}
enter only through $p(S)$.

\begin{proposition}[Time-distributed exposure]
\label{prop:time-dist}
Suppose that if coalition $S$ first leaks at round $t$, then its expected
discounted exposure equals
\[
K_t(S)
=
\sum_{r \in \Rset_S(t)} \delta^{r-t} e_{t,r}(S)
\]
for some nonnegative charges $e_{t,r}(S)$. Then $S$ is deterred if and only if
$K_t(S)\geq\LambdaS_t(S)$ for every $t\leq T$, and honest execution is
coalition-safe if and only if this holds for every $S\in\A$.
\end{proposition}

This covers a fixed charge per compromised round, forfeited future rewards, and
evidence whose detection probability grows with continued reuse;
\Cref{cor:geom-trace} works out a constant per-round tracing hazard. Random
reachability sets and
degraded decoders fit the same comparison
(\Cref{prop:random-reach,prop:modes}).

\paragraph{Long reuse and patient coalitions.}
If a leaked capability stays valid forever and every round is worth at least
$v>0$ to $S$, \Cref{cor:reusable} requires exposure at least $v/(1-\delta)$,
which diverges as $\delta\to 1$. Two designs keep the requirement bounded.
Bounding the lifetime of the leaked object caps the multiplier: when every leak
reaches at most $\rho(S)$ rounds, \Cref{cor:bounded-reuse} asks for at most
$\rho(S)\bar v(S)$ for every $\delta\leq 1$. Liability that grows with use works
from the cost side: a charge $c\geq\sup_r v_r(S)$ per compromised round gives
$K_t(S)=c\sum_{r\in\Rset_S(t)}\delta^{r-t}\geq\LambdaS_t(S)$ at every
$\delta<1$. A fixed one-time loss, reputational or otherwise, does not remove
the divergence.

\section{Instantiation and Applications}
\label{sec:apps}

\subsection{From a Deployment to the Model}
\label{sec:apps-map}

\begin{table}[t]
\centering
\small
\begin{tabularx}{\linewidth}{@{}>{\raggedright\arraybackslash}p{0.12\linewidth}>{\raggedright\arraybackslash}X>{\raggedright\arraybackslash}p{0.40\linewidth}@{}}
\toprule
Input & Source in a deployment & How to bound it \\
\midrule
$\A$ & decryption policy, such as $k$ of $n$ keypers & read off the protocol \\
$\Rset_S(t)$ & scope of the leaked object, custody of key shares, key lifecycle & read off the keying design (\Cref{tab:scheme-mappings}) \\
$v_t(S)$ & surplus that early access to round $t$ creates for a buyer & upper bound: replay order flow with and without early access, plus a tail margin \\
$\kappa(S)$, $K_t(S)$ & slashable stake, tracing probability and penalty, forfeited rewards, legal liability & lower bound \\
$\delta$ & discounting per context & $\delta=1$ over a finite window, or $\delta=e^{-a\Delta}$ for annual rate $a$ and spacing $\Delta$ \\
\bottomrule
\end{tabularx}
\caption{Instantiation map for \Cref{thm:master}.}
\label{tab:instantiation}
\end{table}

\Cref{tab:instantiation} lists where each input of \Cref{thm:master} comes
from; certification needs bounds, not point estimates. For a one-time exposure
the comparison in \Cref{thm:master} is monotone in each input, so an upper bound
on reachable value and a lower bound on exposure certify deterrence; since $\LambdaS_t(S)$ is nondecreasing in $\delta$ and
in $\Rset_S(t)$, a designer unsure of either can certify at an upper bound. A
time-distributed exposure $K_t(S)$ grows with $\delta$ and $\Rset_S(t)$ as well,
so its lower bound must hold over the whole range of $\delta$ and on the
smallest admissible reachability set; a charge of at least $\sup_r v_r(S)$ per
compromised round passes for every $\delta$ and every reachability set.
Competing buyers pay no more than the surplus that early access creates, so
$v_t(S)$ also bounds what a coalition can charge. When no
feasible exposure clears the bound, the remaining levers are the reach of the
leaked object, liability that grows with use (\Cref{sec:beyond}), the value of
early access, and the access structure; under context binding the reach is one
round, so per-use liability is one-time and only the last two remain.

\subsection{Threshold-Encrypted Mempools}
\label{sec:apps-mempool}

In an encrypted mempool the parties are the keypers or validators that hold
key shares, a context is an epoch or a batch, and $v_t(S)$ is what a
front-runner would pay for round $t$'s plaintext order flow before the honest
release point.

\begin{corollary}[Encrypted-mempool instantiation]
\label{cor:mempool}
Consider an encrypted-mempool deployment in which $v_t(S)$ is the maximum
transferable-utility gain available to coalition $S$ from selling premature
access to round $t$'s plaintext order flow. Then $S$ is deterred in a fully
reusable scheme if and only if
${\kappa(S)\geq\sup_{t\leq T}\sum_{r=t}^{T}\delta^{r-t}v_r(S)}$, and in a
context-isolating scheme if and only if ${\kappa(S)\geq\sup_{t\leq T}v_t(S)}$.
\end{corollary}

\Cref{tab:scheme-mappings} maps keying designs to reachability sets, including
two deployed systems. Ferveo runs its distributed key generation again for each
validator epoch, and it allows one key to serve more than one epoch at a cost in
safety \cite{BebelOjha22}. Shutter's keypers generate a long-lived eon key and
derive a key for each epoch from it \cite{DziembowskiFaustLuhn24}. A coalition
that can produce an epoch key also holds a threshold of eon-key shares, so by
\Cref{prop:modes} its exposure must cover the reach of that larger leak, up to
the next eon key generation. Committee rotation affects reachability only
through key continuity: heavy overlap in membership without key reuse leaves
$\Rset_S(t)=\{t\}$, while a key that stays valid across epochs produces
multi-round reuse even with little overlap. Rotation also changes which
coalitions hold a threshold: with fresh keys the model sets $v_t(S)=0$ in epochs
where $S$ holds none, and keys that outlive a committee need time-indexed access
structures (\Cref{sec:open}).

\begin{table}[t]
\centering
\small
\begin{tabularx}{\linewidth}{@{}>{\raggedright\arraybackslash}p{0.38\linewidth}>{\raggedright\arraybackslash}X>{\raggedright\arraybackslash}p{0.25\linewidth}@{}}
\toprule
Deployment pattern & Leaked object & Reachability $\Rset_S(t)$ \\
\midrule
Long-lived threshold key & decoder for the persistent key & $\{t,t+1,\dots,T\}$ \\
Context-dependent threshold decryption \cite{BonehBunzNayakRotemShoup25}, shares in revocable modules & decryption shares bound to $\chi_t$ & $\{t\}$ \\
Context-dependent threshold decryption, other custody of shares & a threshold of key shares, or access to modules holding them & epochs up to the next key generation \\
Rotating committees, fresh key per epoch & decoder for epoch $t$'s key only & $\{t\}$ \\
Windowed key reuse or lagged refresh & decoder valid for $w$ epochs & $\subseteq\{t,\dots,t+w-1\}$ \\
Ferveo \cite{BebelOjha22}, new key each epoch & decryption capability for epoch $t$ & $\{t\}$ \\
Ferveo, one key for $L$ epochs & capability for the shared key & epochs up to the next key generation, at most $L$ \\
Shutter \cite{DziembowskiFaustLuhn24}, epoch key leaked & secret key of epoch $t$ & $\{t\}$ \\
Shutter, eon key shares leaked & a threshold of eon key shares, which yield every epoch key & epochs up to the next eon key generation \\
\bottomrule
\end{tabularx}
\caption{From keying designs to reachability sets. When a coalition can leak
several objects, the most valuable one governs (\Cref{prop:modes}).
Accountability mechanisms that only penalize, such as traitor tracing
\cite{BonehPartapRotem23} or accountable logging such as L{\O}
\cite{NasrulinIshmaevDecouchantPouwelse23}, change $\kappa(S)$ or $K_t(S)$ and
leave these sets unchanged; revocation or rekeying on detection shrinks them.}
\label{tab:scheme-mappings}
\end{table}

\begin{table}[t]
\centering
\small
\begin{tabular}{@{}lrrr@{}}
\toprule
Per-attack profit $v$ & $\rho=1$ (context isolation) & $\rho=4$ & $\rho=32$ (full reuse) \\
\midrule
Median & \$24.07 & \$96.28 & \$770.24 \\
75th percentile & \$62.92 & \$251.68 & \$2{,}013.44 \\
Maximum & \$20{,}084.01 & \$80{,}336.04 & \$642{,}688.32 \\
\bottomrule
\end{tabular}
\caption{Required exposure $\rho v$ when the per-context value $v$ is set to the
Ethereum insertion-attack profits of Torres et~al.~\cite{TorresCaminoState21}, with
$\delta=1$ over a key lifetime of $32$ block contexts.}
\label{tab:torres}
\end{table}

\begin{example}[Calibration worksheet]
\label{ex:mempool-num}
Suppose a conservative cap on one epoch's leakage value is \$50{,}000 and
$\delta=0.95$. If a leaked decoder stays valid forever, \Cref{cor:reusable}
requires $\kappa(S)\geq 50{,}000/(1-0.95)=1{,}000{,}000$. Context isolation
reduces this to $50{,}000$ by \Cref{cor:context}, and a reuse cap of
$\rho(S)=4$ epochs gives $50{,}000(1+0.95+0.95^2+0.95^3)\approx 185{,}494$ by
\Cref{cor:bounded-reuse}. Tracing that succeeds with probability at least $0.4$,
against a fine of \$500{,}000 after a trace, gives expected exposure $200{,}000$
by \Cref{prop:prob}. That budget deters context-isolated leakage with room to
spare, clears the four-epoch requirement narrowly, and falls five times short of
full reuse.
\end{example}

The dollar inputs of \Cref{ex:mempool-num} are illustrative, and
\Cref{tab:torres} anchors the scale in data. Torres et~al.~\cite{TorresCaminoState21}
identify 196{,}691 insertion (sandwich) attacks on Ethereum, with per-attack
profit of \$24.07 at the median, \$62.92 at the 75th percentile, and
\$20{,}084.01 at the maximum (their Table~3). We take a key reused over $32$
block contexts, set $\delta=1$ over this short window, and use a constant stage
value $v$; a leak that reaches $\rho$ contexts then requires
$\kappa(S)\geq\rho v$. Observed attacker profit is a stress benchmark for
$v_t(S)$ and not a certified upper bound, so a deployment certificate should
aggregate every profitable response to early access and add tail and estimation
margins.

\subsection{Other Applications}
\label{sec:apps-other}

\paragraph{Federated learning under a collective key.}
In POSEIDON \cite{SavPyrgelisTroncosoPastorizaFroelicherBossuatSousaHubaux21},
parties train a neural network under the collective public key of a multiparty
homomorphic encryption scheme \cite{MouchetTroncosoPastorizaBossuatHubaux21},
and decryption requires their cooperation. A training round is a context: a
coalition able to decrypt can sell a buyer the round's model or aggregated
updates, and $v_t(S)$ is the value of that access. If the collective key
persists across rounds, one leak reaches every later round, and refreshing the
key per round or per window of rounds bounds the reachability set as in
\Cref{cor:bounded-reuse}. MedCo
\cite{RaisaroTroncosoPastorizaMisbachSousaPradervandMissiagliaMichielinFordHubaux19}
applies the same collective-key design to queries on distributed clinical and
genomic data, with a query as the context.

\paragraph{Sealed-bid auctions and voting.}
Sealed-bid auction services keep bids hidden among several servers until the
auction closes \cite{FranklinReiter96}. When bids are encrypted to a threshold
key held by the servers, a coalition that leaks the key's decryption capability
opens the bids of every later auction under that key; bids shared afresh for
each auction give $\Rset_S(t)=\{t\}$. In the election scheme of
Cramer et~al.~\cite{CramerGennaroSchoenmakers97}, authorities share a threshold key and
decrypt only the tally; a quorum can instead sell individual ballots to a vote
buyer, and a key reused across elections lets one leak reach later elections.
Ballots are never released, so the leak also opens past elections, which
reachability sets do not count; for data that are never released, the threshold
of \Cref{thm:master} is a lower bound on the exposure needed.

\paragraph{Threshold signatures.}
The first-leak argument also covers threshold signatures, where the object sold
is the ability to act. Then $v_t(S)$ is the total gain from unauthorized
signing in context $t$, and $\Rset_S(t)$ records the contexts in which a leaked
signing capability stays valid: every later epoch for a long-lived bridge key,
even when each signature is domain separated, and a single epoch only for
independent per-epoch keys or for keys kept in modules that sign only for the
current epoch and are replaced as soon as a leak becomes public.
\Cref{thm:master} applies when forged actions have fixed reachable values; a
forged action that changes balances, nonces, or the set of feasible later
actions makes values state dependent, which the model does not cover.

\section{Limitations and Scope}
\label{sec:limits}

The inputs $v_t(S)$, $\Rset_S(t)$, and $\kappa(S)$ are exogenous, and
\Cref{sec:apps-map} shows how to bound them.

\Cref{thm:master} is exact under transferable utility. When the first-leak
benefit and the exposure are deterministic, any increasing coalition utility
ranks a deviation by the sign of $\LambdaS_t(S)-\kappa(S)$, so risk aversion
leaves the threshold unchanged. Risk attitudes matter once exposure or value is
random: the expected-exposure reductions of \Cref{prop:prob,prop:time-dist} are
exact under risk neutrality and with detection that does not respond to later
strategic behavior. Without side transfers a coalition has no scalar objective,
and the threshold is a transferable-utility benchmark for it.

The benchmark is a leak by one coalition while every other party executes
honestly. When authorized coalitions overlap, one coalition's leak can preempt
another's; \Cref{rem:race} gives the comparison under exogenous survival
probabilities, and the full race is open. The model also takes a leak to be
terminal (\Cref{sec:reach}); when tracing fails to identify the leakers this
rests on that assumption alone, since unidentified leakers would receive any
replacement keys or modules. A coalition that could keep leaking after a single charge would
collect every later round under any keying design, so the fully reusable
threshold would apply even under context binding.

Deterring early leakage does not make execution fair after decryption.
\Cref{prop:ordering} gives a context-isolating deployment in which no coalition
leaks and a block builder still sandwiches a user after honest decryption, so
ordering after release needs its own mechanism or consensus constraints, and a
theory that combines the two layers is open.

\section{Open Problems}
\label{sec:open}

\begin{enumerate}[leftmargin=2em]
\item \textbf{Endogenous leakage markets.} Derive the stage value $v_t(S)$ from
strategic buyers, buyer competition, limited commitment, and competing
coalitions.

\item \textbf{Dynamic and rotating committees.} Extend the reachability
characterization to time-indexed access structures $\A_t$, time-indexed
coalitions $S_t$, and partially overlapping committees.

\item \textbf{Partial context isolation.} \Cref{cor:bounded-reuse} prices reuse
multiplicity alone; exact thresholds for leaked objects that work on a
structured subset of future contexts remain open.

\item \textbf{Beyond transferable utility.} Which thresholds survive imperfect
side transfers, endogenous tracing investment, or risk aversion under random
exposure?
\end{enumerate}

\begin{ack}
This work received no third-party funding. The authors declare no competing
interests. The views expressed in this paper are those of the authors and do
not necessarily reflect the views of Tropic.
\end{ack}

\clearpage
\appendix

\section{Proofs for the Main Results}
\label{sec:proofs}

This appendix proves the results of \Cref{sec:main} in the order in which
they are stated.

\subsection{Normal Form for Coalition Deviations}

The first lemma reduces every coalition deviation to its first successful leak,
which turns a repeated deviation into a supremum over first-leak times.

\begin{lemma}[Earliest-leak normal form]
\label{lem:normal}
Fix a coalition $S \subseteq \Players$, a round $t_0$, and a public history with
no successful leak before $t_0$. Consider a pure deviation strategy by $S$ from
$t_0$ on, and let $\tau\geq t_0$ be the round of its successful leak, with
$\tau=\infty$ if it never leaks. In present value at round $t_0$ and relative
to remaining honest forever, the deviation yields
\[
\begin{cases}
\delta^{\tau-t_0}\bigl(\LambdaS_\tau(S)-\kappa(S)\bigr), & \text{if } \tau < \infty,\\
0, & \text{if } \tau = \infty.
\end{cases}
\]
\end{lemma}

\begin{proof}
We treat the two cases in turn. If $\tau=\infty$, the coalition never leaks,
so under the payoff convention of \Cref{sec:model} it gains and pays nothing.

Suppose $\tau<\infty$. No round before $\tau$ is compromised, because $\tau$ is
the only leak. The leak is terminal, so the buyer reads exactly the rounds in
$\Rset_S(\tau)$, each worth its stage value, and the exposure $\kappa(S)$ is
charged once, at $\tau$. Discounting every amount to $t_0$ gives
\[
\sum_{r\in\Rset_S(\tau)}\delta^{r-t_0}v_r(S)-\delta^{\tau-t_0}\kappa(S)
=\delta^{\tau-t_0}\bigl(\LambdaS_\tau(S)-\kappa(S)\bigr).
\]
\end{proof}

\subsection{Exact Characterization}

\begin{proof}[Proof of \Cref{thm:master}]
We fix a coalition $S$ and prove sufficiency first. Suppose that
$\kappa(S)\geq\sup_{t\leq T}\LambdaS_t(S)$, and take a round $t_0$, a public
history with no leak before $t_0$, and a pure deviation strategy. By
\Cref{lem:normal} the deviation yields $0$ or
$\delta^{\tau-t_0}(\LambdaS_\tau(S)-\kappa(S))$, and the assumed inequality
makes the latter nonpositive. A randomized deviation averages these values, so
$S$ is deterred.

For necessity, suppose that $\kappa(S)<\sup_{t\leq T}\LambdaS_t(S)$. By the
definition of the supremum we can choose a round $t^\star$ with
$\LambdaS_{t^\star}(S)>\kappa(S)$, whether or not the supremum is attained. We
let $S$ stay honest until round $t^\star$ and leak there. Evaluated at round
$1$, this deviation yields, by \Cref{lem:normal},
\[
\delta^{t^\star-1}\bigl(\LambdaS_{t^\star}(S)-\kappa(S)\bigr)>0,
\]
so $S$ is not deterred. For an unauthorized coalition every reachability set is
empty, so $\sup_{t}\LambdaS_t(S)=0\leq\kappa(S)$.
\end{proof}

\begin{proof}[Proof of \Cref{cor:oneshot}]
With $T=1$ we have $\Rset_S(1)=\{1\}$ for every authorized $S$, because
$\Rset_S(1)\subseteq\{1\}$ contains round $1$. So $\LambdaS_1(S)=v_1(S)$, and
\Cref{thm:master} gives the claim.
\end{proof}

\begin{proof}[Proof of \Cref{cor:reusable}]
We substitute $\Rset_S(t)=\{t,\dots,T\}$ into \Cref{def:lambda}, which gives
$\LambdaS_t(S)=\sum_{r=t}^{T}\delta^{r-t}v_r(S)$ for every authorized $S$, and
\Cref{thm:master} gives the first statement. If in addition $T=\infty$,
$\delta<1$, and $v_r(S)=v(S)$ for all $r$, the geometric series gives, for every
$t$,
\[
\LambdaS_t(S)=\sum_{j=0}^{\infty}\delta^j v(S)=\frac{v(S)}{1-\delta}.
\]
\end{proof}

\begin{proof}[Proof of \Cref{cor:context}]
Under context isolation we have $\LambdaS_t(S)=v_t(S)$ for every authorized $S$
and round $t$. \Cref{thm:master} gives the claim.
\end{proof}

\subsection{Separation and Converse Lower Bound}

\begin{proof}[Proof of \Cref{thm:separation}]
Under fully reusable leakage with $\delta=1$ and unit values,
\[
\LambdaS_t(S^\star)=\sum_{r=t}^{T} 1 = T-t+1,
\]
which is largest at $t=1$, so $\sup_{t\leq T}\LambdaS_t(S^\star)=T$ and
\Cref{thm:master} gives (i). Under context isolation
$\LambdaS_t(S^\star)=v_t(S^\star)=1$ for every $t$, and \Cref{thm:master} gives
(ii). For the last claim we fix a $T$-round instance and a coalition $S$. Since
$\delta^{r-t}\leq 1$ and $v_r(S)\leq\sup_{t'\leq T}v_{t'}(S)$, every round
$t\leq T$ satisfies
\[
\sum_{r=t}^{T}\delta^{r-t}v_r(S)\leq(T-t+1)\sup_{t'\leq T}v_{t'}(S)
\leq T\sup_{t'\leq T}v_{t'}(S).
\]
\end{proof}

\begin{proof}[Proof of \Cref{prop:reuse-mult}]
Since $\delta=1$ and every stage value of $S$ is $1$, each reachable round
contributes exactly $1$, and we get
\[
\LambdaS_t(S)=\sum_{r\in\Rset_S(t)}1=|\Rset_S(t)|
\qquad\text{for every round } t.
\]
So $\sup_{t\leq T}\LambdaS_t(S)=\rho(S)$, and \Cref{thm:master} makes $\rho(S)$
the threshold for $S$. Along a family in which $\rho(S)$ grows without bound,
this threshold exceeds every constant.
\end{proof}

\begin{proof}[Proof of \Cref{cor:bounded-reuse}]
We fix a round $t$. Since $v_r(S)\leq \bar v(S)$ for every round $r$ and the
discount weights $\delta^{r-t}$ weakly decrease in the delay $r-t$, a discounted
sum over $|\Rset_S(t)|$ rounds from $t$ on is largest when those rounds come
first. Because $|\Rset_S(t)|\leq\rho(S)$, this gives
\[
\LambdaS_t(S)
\leq
\bar v(S)\sum_{j=0}^{|\Rset_S(t)|-1}\delta^j
\leq
\bar v(S)\sum_{j=0}^{\rho(S)-1}\delta^j,
\]
and \Cref{thm:master} shows that $S$ is deterred. For tightness, suppose that
$\Rset_S(t)=\{t,t+1,\dots,t+\rho(S)-1\}$ for some round $t$ and that each of
these rounds has stage value $\bar v(S)$. Both inequalities are then equalities
at that $t$, so \Cref{thm:master} rules out any smaller exposure.
\end{proof}

\subsection{Access-Structure Corollaries}

\begin{proof}[Proof of \Cref{cor:uniform}]
Since $1\leq k\leq n$, the family $\A$ is finite and nonempty and every
$S\in\A$ has $|S|\geq 1$, so we may divide by $|S|$. By \Cref{thm:master},
honest execution is then coalition-safe exactly when
\[
d \geq \frac{\sup_{t \leq T}\LambdaS_t(S)}{|S|}
\qquad\text{for every } S \in \A,
\]
and the smallest such $d$ is the maximum of the right-hand side over $S\in\A$,
which is $d^\star$; it is finite by the standing assumption of
\Cref{def:lambda}. When the numerator depends only on $|S|$, coalitions of equal
size impose the same constraint, and the maximum runs over sizes:
\[
d^\star=\max_{k\leq m\leq n}\frac{\lambda_m}{m}.
\]
\end{proof}

\begin{proof}[Proof of \Cref{cor:hetero-lp}]
We apply \Cref{thm:master} to the additive exposure
$\kappa(S)=\sum_{i\in S}d_i$: the deterring profiles are exactly the nonnegative
solutions of the covering constraints
\[
\sum_{i \in S} d_i \geq \sup_{t \leq T}\LambdaS_t(S)
\qquad\text{for every } S \in \A,
\]
whose right-hand sides are finite by the standing assumption of
\Cref{def:lambda}. Minimizing $\sum_i d_i$ over them is a linear program.
\end{proof}

\section{Extensions, Examples, and Boundary Results}
\label{sec:boundary}

\subsection{Proofs of the Exposure Extensions}
\label{sec:boundary-prob}

\begin{proof}[Proof of \Cref{prop:prob}]
We compare with the baseline model in which $S$ bears the deterministic
exposure $\kappa'(S)=p(S)f(S)$. The detection probability and the fine are
fixed in advance, so a risk-neutral coalition values a deviation by its
expected payoff, and the proof of \Cref{lem:normal} values a pure deviation
whose leak comes at round $\tau$ at
$\delta^{\tau-t_0}(\LambdaS_\tau(S)-p(S)f(S))$ in expectation at round $t_0$,
its value in that baseline. \Cref{thm:master} for $\kappa'$ gives the claim.
\end{proof}

\begin{proof}[Proof of \Cref{prop:time-dist}]
We fix a coalition $S$, a round $t_0$, and a public history with no leak before
$t_0$. The proof of \Cref{lem:normal}, with the expected charge $K_\tau(S)$ in
place of $\kappa(S)$, values a pure deviation whose leak comes at round $\tau$
at $\delta^{\tau-t_0}(\LambdaS_\tau(S)-K_\tau(S))$ and one that never leaks at
$0$. If $K_t(S)\geq\LambdaS_t(S)$ for every $t$, these values and their
averages are nonpositive, so $S$ is deterred. If $K_t(S)<\LambdaS_t(S)$ for
some $t$, staying honest until $t$ and leaking there gains
$\delta^{t-1}(\LambdaS_t(S)-K_t(S))>0$ from round $1$. Unauthorized coalitions
have $K_t(S)=\LambdaS_t(S)=0$.
\end{proof}

\subsection{Geometric Evidence Accumulation}
\label{sec:boundary-dynamic}

Under a constant per-round tracing hazard, the time-distributed exposure of
\Cref{prop:time-dist} has a closed form.

\begin{corollary}[Geometric evidence accumulation]
\label{cor:geom-trace}
Fix an authorized coalition $S$ and let $T=\infty$. Suppose a leak at round $t$
compromises exactly the next $\rho(S)$ rounds, so
$\Rset_S(t)=\{t,t+1,\dots,t+\rho(S)-1\}$, where $\rho(S)$ may be finite or
$\infty$. Suppose further that each compromised round independently produces
fresh tracing evidence with hazard $q(S)\in[0,1]$, and that the coalition pays
the fine $f(S)$ upon the first traced round. Then
\[
K_t(S)
=
\sum_{j=0}^{\rho(S)-1}f(S)\,q(S)\bigl(\delta(1-q(S))\bigr)^j.
\]
If $\rho(S)<\infty$ and $\delta(1-q(S))\neq 1$, this equals
\[
K_t(S)
=
f(S)\,q(S)\,
\frac{1-\bigl(\delta(1-q(S))\bigr)^{\rho(S)}}{1-\delta(1-q(S))},
\]
and if $\rho(S)=\infty$ and $\delta(1-q(S))<1$, then
\[
K_t(S)
=
\frac{f(S)\,q(S)}{1-\delta(1-q(S))}.
\]
Hence $S$ is deterred if and only if this quantity dominates $\LambdaS_t(S)$
for every round $t$.
\end{corollary}

\begin{proof}
We apply \Cref{prop:time-dist} with $e_{t,t+j}(S)=f(S)\,q(S)(1-q(S))^j$ for
each reachable offset $j$: the fine $f(S)$ times the probability that the first
trace comes $j$ rounds after the leak. The first formula is the resulting sum.
Its finite and infinite geometric sums give the closed forms.
\end{proof}

\subsection{Stochastic Reuse and Partial Leakage}

The same first-leak reduction continues to work if reuse or leakage quality is
random, provided the coalition evaluates deviations in expected money terms.

\begin{proposition}[Stochastic reachability]
\label{prop:random-reach}
Suppose that if coalition $S$ leaks at round $t$, then the compromised rounds
form a random set $\widetilde{\Rset}_S(t) \subseteq \{t,\dots,T\}$ whose law
does not depend on the public history before round $t$. Define the expected
reachable leakage value
\[
\overline{\Lambda}_t(S)
=
\mathbb{E}\!\left[
\sum_{r \in \widetilde{\Rset}_S(t)} \delta^{r-t} v_r(S)
\right],
\]
and assume that it is finite. If the coalition is risk neutral and the expected
discounted exposure from a leak at round $t$ is $K_t(S)$, then honest execution
is coalition-safe if and only if
\[
K_t(S)\geq\overline{\Lambda}_t(S)
\qquad\text{for every } S \in \A \text{ and every } t \leq T.
\]
\end{proposition}

\begin{proof}
Since the law of $\widetilde{\Rset}_S(t)$ does not depend on the public history,
the expected benefit of a leak at round $t$ is $\overline{\Lambda}_t(S)$ after
every history with no earlier leak. We therefore repeat the proof of
\Cref{prop:time-dist} with $\overline{\Lambda}_t(S)$ in place of
$\LambdaS_t(S)$: a pure deviation whose leak comes at round $\tau$ is worth in
expectation, at round $t_0$,
\[
\delta^{\tau-t_0}\bigl(\overline{\Lambda}_\tau(S)-K_\tau(S)\bigr).
\]
\end{proof}

\begin{remark}
\Cref{prop:random-reach} covers deployments in which reuse ends at a random
time drawn after the leak, such as revocation triggered by stochastic evidence
or a memoryless refresh schedule. A refresh schedule that the public history
reveals, such as a fixed period with a public phase, belongs to the
deterministic model, with $\Rset_S(t)$ read off the schedule.
\end{remark}

\begin{proposition}[Leakage modes and degraded decoders]
\label{prop:modes}
For each coalition $S$, let $\mathcal{L}_S$ be a set of possible leak modes,
empty when $S\notin\A$. If mode $\ell \in \mathcal{L}_S$ is chosen at round $t$,
it generates stage values $v_r^{\ell}(S)$ over a reachable set
$\Rset_S^{\ell}(t) \subseteq \{t,\dots,T\}$. Define
\[
\LambdaS_t^{\ell}(S)
=
\sum_{r \in \Rset_S^{\ell}(t)} \delta^{r-t} v_r^{\ell}(S),
\]
with a finite supremum over $\ell$ and $t$ for every $S\in\A$. Then an
authorized coalition $S$ is deterred if and only if
\[
\kappa(S)\geq\sup_{\ell \in \mathcal{L}_S,\ t \leq T}\LambdaS_t^{\ell}(S),
\]
and honest execution is coalition-safe if and only if this holds for every
$S\in\A$.
\end{proposition}

\begin{proof}
A leak now consists of a round $\tau$ and a mode $\ell$, and the proof of
\Cref{lem:normal} values a pure deviation that leaks in mode $\ell$ at round
$\tau$ at $\delta^{\tau-t_0}(\LambdaS^{\ell}_\tau(S)-\kappa(S))$, evaluated at
round $t_0$. An unauthorized coalition has no mode and never leaks. We then
repeat the proof of \Cref{thm:master} with the supremum taken over both
choices.
\end{proof}

\begin{remark}
A degraded decoder that is useful only for a subset of transactions is simply
one leakage mode $\ell$: its lower coverage appears as smaller stage values
$v_r^{\ell}(S)$, and its weaker portability appears as a smaller reachable set
$\Rset_S^{\ell}(t)$. The proposition needs no monotonicity or subadditivity
assumptions on those objects.
\end{remark}

\subsection{An Explicit \texorpdfstring{$k$-out-of-$n$}{k-out-of-n} Example}

\begin{example}[Stationary reusable versus context-bound committees]
\label{ex:ton}
Let $\A=\{S \subseteq \Players: |S|\geq k\}$ for some $1\leq k\leq n$, let
$T=\infty$ and $\delta \in (0,1)$, and assume that every authorized coalition of
size $m$ has constant per-round leakage value $v(m)$. In a fully reusable
deployment,
\[
\sup_t \LambdaS_t(S)=\frac{v(|S|)}{1-\delta},
\]
so by \Cref{cor:uniform} the exact equal-exposure requirement is
\[
d^\star_{\mathrm{reuse}}
=
\max_{k\leq m\leq n}\frac{v(m)}{m(1-\delta)}.
\]
In a context-isolating deployment, the factor $1/(1-\delta)$ disappears:
\[
d^\star_{\mathrm{ctx}}
=
\max_{k\leq m\leq n}\frac{v(m)}{m}.
\]
Thus context isolation removes the entire continuation-value multiplier.
\end{example}

\subsection{Proof of the Encrypted-Mempool Corollary}

\begin{proof}[Proof of \Cref{cor:mempool}]
Under full reuse $\LambdaS_t(S)=\sum_{r=t}^{T}\delta^{r-t}v_r(S)$, and under
context isolation $\LambdaS_t(S)=v_t(S)$. We apply \Cref{thm:master} in each
case.
\end{proof}

\subsection{What the Model Does Not Solve}
\label{sec:boundary-ordering}

The paper's guarantee concerns leakage before the honest release point. It does
not imply fairness after plaintexts become public.

\begin{proposition}[Early-leakage deterrence does not imply order fairness]
\label{prop:ordering}
There exists a deployment in which honest execution is coalition-safe under
the exact threshold of \Cref{thm:master}, yet a block builder can still extract
positive value after honest decryption by placing its own transactions around a
user's.
\end{proposition}

\begin{proof}
We take a context-isolating deployment whose exposure satisfies
\Cref{cor:context}, so that every coalition is deterred, and we let one round
contain a user order that sells $10$ units of token $X$ to a constant-product
pool with reserves $(100,100)$ and no fee, with a minimum output of $8$ units of
token $Y$. Every trade preserves the product $10^4$ of the reserves. After
honest decryption the builder first sells $5$ units of $X$ and receives $100/21$
units of $Y$. The user's order still clears, since it now receives $4000/483>8$
units of $Y$, and the builder then sells its $100/21$ units of $Y$ back for
$2645/443$ units of $X$, a profit of
\[
\frac{2645}{443}-5=\frac{430}{443}>0.
\]
\end{proof}

\begin{remark}[Race-to-leak interactions]
\label{rem:race}
\Cref{thm:master} is a baseline first-leak deterrence statement. A lightweight
extension gives coalition $S$ exogenous survival probabilities
$\sigma_{t,r}(S)\in[0,1]$, where $\sigma_{t,r}(S)$ is the probability that
$S$'s buyer still enjoys exclusive early access in round $r$ if $S$ leaks at
round $t$. Suppose that $S$ is risk neutral, that access without exclusivity is
worth nothing, and that the random set of rounds at which exclusivity survives
has a law that does not depend on the public history. Taking that set as the
random reachable set, \Cref{prop:random-reach} with $K_t(S)=\kappa(S)$ gives the
comparison $\kappa(S)\geq\sup_t\widetilde{\Lambda}_t(S)$, where
\[
\widetilde{\Lambda}_t(S)
=
\sum_{r \in \Rset_S(t)} \delta^{r-t} \sigma_{t,r}(S) v_r(S).
\]
The hard part is endogenizing the survival probabilities when authorized
coalitions overlap and strategically preempt one another. That full race-to-leak equilibrium problem remains
outside the present paper.
\end{remark}

\newpage
\section*{NeurIPS Paper Checklist}

\begin{enumerate}

\item {\bf Claims}
    \item[] Question: Do the main claims made in the abstract and introduction accurately reflect the paper's contributions and scope?
    \item[] Answer: \answerYes{}.
    \item[] Justification: The abstract and introduction state the theorem-level contributions and delimit the model as a stylized theoretical analysis. The main claims are matched to formal definitions, theorem statements, corollaries, and proof sections in the body of the paper.

\item {\bf Limitations}
    \item[] Question: Does the paper discuss the limitations of the work performed by the authors?
    \item[] Answer: \answerYes{}.
    \item[] Justification: The paper discusses its modeling scope, boundary cases, and open problems in the introduction, model discussion, boundary/extension discussion, and concluding open-problems material. In particular, the analysis is presented as a theoretical characterization rather than an empirical or deployed-system evaluation. Section 7 (Limitations and Scope) collects these limitations in one place.

\item {\bf Theory assumptions and proofs}
    \item[] Question: For each theoretical result, does the paper provide the full set of assumptions and a complete (and correct) proof?
    \item[] Answer: \answerYes{}.
    \item[] Justification: The paper states the assumptions through formal definitions and theorem hypotheses, and provides proofs for the main theorems, lemmas, propositions, and corollaries. Results explicitly marked as sketches or boundary observations are identified as such.

    \item {\bf Experimental result reproducibility}
    \item[] Question: Does the paper fully disclose all the information needed to reproduce the main experimental results of the paper to the extent that it affects the main claims and/or conclusions of the paper (regardless of whether the code and data are provided or not)?
    \item[] Answer: \answerNA{}.
    \item[] Justification: The paper does not report experiments. The main results are mathematical statements whose verification depends on the formal model, theorem statements, and proofs provided in the paper.

\item {\bf Open access to data and code}
    \item[] Question: Does the paper provide open access to the data and code, with sufficient instructions to faithfully reproduce the main experimental results, as described in supplemental material?
    \item[] Answer: \answerNA{}.
    \item[] Justification: The paper does not use datasets, trained models, or experimental code to support its claims. Algorithmic statements, where present, are specified and analyzed mathematically in the text.

\item {\bf Experimental setting/details}
    \item[] Question: Does the paper specify all the training and test details (e.g., data splits, hyperparameters, how they were chosen, type of optimizer) necessary to understand the results?
    \item[] Answer: \answerNA{}.
    \item[] Justification: The paper contains no training, test set, hyperparameter, optimizer, or empirical evaluation setup. Its results are derived from the stated theoretical model.

\item {\bf Experiment statistical significance}
    \item[] Question: Does the paper report error bars suitably and correctly defined or other appropriate information about the statistical significance of the experiments?
    \item[] Answer: \answerNA{}.
    \item[] Justification: The paper does not include experiments or statistical estimates. Consequently, there are no empirical error bars, confidence intervals, or significance tests to report.

\item {\bf Experiments compute resources}
    \item[] Question: For each experiment, does the paper provide sufficient information on the computer resources (type of compute workers, memory, time of execution) needed to reproduce the experiments?
    \item[] Answer: \answerNA{}.
    \item[] Justification: The paper does not run computational experiments. No experimental compute resources are required to reproduce the paper's main claims.
    
\item {\bf Code of ethics}
    \item[] Question: Does the research conducted in the paper conform, in every respect, with the NeurIPS Code of Ethics \url{https://neurips.cc/public/EthicsGuidelines}?
    \item[] Answer: \answerYes{}.
    \item[] Justification: The work is a theoretical analysis and does not involve human subjects, private data, deployed interventions, or release of potentially harmful models or datasets. The paper preserves anonymity in the submission version.

\item {\bf Broader impacts}
    \item[] Question: Does the paper discuss both potential positive societal impacts and negative societal impacts of the work performed?
    \item[] Answer: \answerYes{}.
    \item[] Justification: The paper discusses positive implications for robust, privacy-preserving, or incentive-compatible decentralized mechanisms, while also identifying boundary cases and residual attack surfaces. The work is theoretical and does not propose a deployed system, so the impact discussion is correspondingly scoped to technical security and market-design implications.
    
\item {\bf Safeguards}
    \item[] Question: Does the paper describe safeguards that have been put in place for responsible release of data or models that have a high risk for misuse (e.g., pre-trained language models, image generators, or scraped datasets)?
    \item[] Answer: \answerNA{}.
    \item[] Justification: The paper does not release data, trained models, scraped datasets, or other assets with high misuse risk. No release-specific safeguards are therefore applicable.

\item {\bf Licenses for existing assets}
    \item[] Question: Are the creators or original owners of assets (e.g., code, data, models), used in the paper, properly credited and are the license and terms of use explicitly mentioned and properly respected?
    \item[] Answer: \answerNA{}.
    \item[] Justification: The paper does not use existing code, datasets, models, benchmarks, or other external assets as research inputs. Prior scholarly work is credited through citations in the related-work and references sections.

\item {\bf New assets}
    \item[] Question: Are new assets introduced in the paper well documented and is the documentation provided alongside the assets?
    \item[] Answer: \answerNA{}.
    \item[] Justification: The paper does not introduce or release new datasets, code packages, models, or benchmarks. Its contribution consists of formal models, theorems, and proofs.

\item {\bf Crowdsourcing and research with human subjects}
    \item[] Question: For crowdsourcing experiments and research with human subjects, does the paper include the full text of instructions given to participants and screenshots, if applicable, as well as details about compensation (if any)? 
    \item[] Answer: \answerNA{}.
    \item[] Justification: The paper does not involve crowdsourcing, surveys, experiments with human participants, or human-subject data collection. There are therefore no participant instructions, screenshots, or compensation details to report.

\item {\bf Institutional review board (IRB) approvals or equivalent for research with human subjects}
    \item[] Question: Does the paper describe potential risks incurred by study participants, whether such risks were disclosed to the subjects, and whether Institutional Review Board (IRB) approvals (or an equivalent approval/review based on the requirements of your country or institution) were obtained?
    \item[] Answer: \answerNA{}.
    \item[] Justification: The paper does not involve human subjects or crowdsourced participants. IRB or equivalent human-subjects review is therefore not applicable.

\item {\bf Declaration of LLM usage}
    \item[] Question: Does the paper describe the usage of LLMs if it is an important, original, or non-standard component of the core methods in this research? Note that if the LLM is used only for writing, editing, or formatting purposes and does \emph{not} impact the core methodology, scientific rigor, or originality of the research, declaration is not required.
    \item[] Answer: \answerYes{}.
    \item[] Justification: AI tools were used to help check the manuscript's derivations and proofs, and to provide suggestions for improving and revising the derivations.

\end{enumerate}


\begin{thebibliography}{99}

\bibitem{AbrahamDolevGonenHalpern06}
Ittai Abraham, Danny Dolev, Rica Gonen, and Joseph Y. Halpern.
\newblock Distributed Computing Meets Game Theory: Robust Mechanisms for
Rational Secret Sharing and Multiparty Computation.
\newblock In \emph{Proceedings of the Twenty-Fifth Annual ACM Symposium on
Principles of Distributed Computing}, pages 53--62, 2006.

\bibitem{AlipanahlooHafidZhang24}
Zeinab Alipanahloo, Abdelhakim Senhaji Hafid, and Kaiwen Zhang.
\newblock Maximum Extractable Value ({MEV}) Mitigation Approaches in {E}thereum
and Layer-2 Chains: A Comprehensive Survey.
\newblock \emph{IEEE Access}, 12:185212--185231, 2024.

\bibitem{Aumann59}
Robert~J. Aumann.
\newblock Acceptable Points in General Cooperative $n$-Person Games.
\newblock In A.~W. Tucker and R.~D. Luce, editors, \emph{Contributions to the
Theory of Games {IV}}, volume~40 of \emph{Annals of Mathematics Studies}, pages
287--324. Princeton University Press, 1959.

\bibitem{BabelJeanLouisJiMisraKelkarMudiyanselageMillerJuels24}
Kushal Babel, Nerla Jean-Louis, Yan Ji, Ujval Misra, Mahimna Kelkar, Kosala
Yapa Mudiyanselage, Andrew Miller, and Ari Juels.
\newblock {PROF}: Protected Order Flow in a Profit-Seeking World.
\newblock In \emph{2026 IEEE 11th European Symposium on Security and Privacy
(EuroS\&P)}, pages 398--418, 2026.

\bibitem{BebelOjha22}
Joseph Bebel and Dev Ojha.
\newblock Ferveo: Threshold Decryption for Mempool Privacy in {BFT} Networks.
\newblock \emph{Cryptology ePrint Archive}, Paper 2022/898, 2022.

\bibitem{Becker68}
Gary~S. Becker.
\newblock Crime and Punishment: An Economic Approach.
\newblock \emph{Journal of Political Economy}, 76(2):169--217, 1968.

\bibitem{BernheimPelegWhinston87}
B.~Douglas Bernheim, Bezalel Peleg, and Michael~D. Whinston.
\newblock Coalition-Proof {N}ash Equilibria {I}. {C}oncepts.
\newblock \emph{Journal of Economic Theory}, 42(1):1--12, 1987.

\bibitem{BonehBunzNayakRotemShoup25}
Dan Boneh, Benedikt B\"unz, Kartik Nayak, Lior Rotem, and Victor Shoup.
\newblock Context-Dependent Threshold Decryption and Its Applications.
\newblock In \emph{Advances in Cryptology -- {ASIACRYPT} 2025}, Lecture Notes in
Computer Science, pages 506--538. Springer, 2025.

\bibitem{BonehPartapRotem23}
Dan Boneh, Aditi Partap, and Lior Rotem.
\newblock Accountability for Misbehavior in Threshold Decryption via Threshold
Traitor Tracing.
\newblock In \emph{Advances in Cryptology -- {CRYPTO} 2024}, volume 14926 of
\emph{Lecture Notes in Computer Science}, pages 317--351. Springer, 2024.

\bibitem{BudishLewisPyeRoughgarden24}
Eric Budish, Andrew Lewis-Pye, and Tim Roughgarden.
\newblock The Economic Limits of Permissionless Consensus.
\newblock In \emph{Proceedings of the 25th {ACM} Conference on Economics and
Computation}, 2024.

\bibitem{ButerinGriffith17}
Vitalik Buterin and Virgil Griffith.
\newblock Casper the Friendly Finality Gadget.
\newblock \emph{arXiv preprint arXiv:1710.09437}, 2017.

\bibitem{ChenZhuKandasamy23}
Yiding Chen, Jerry Zhu, and Kirthevasan Kandasamy.
\newblock Mechanism Design for Collaborative Normal Mean Estimation.
\newblock In \emph{Advances in Neural Information Processing Systems 36}, 2023.

\bibitem{CramerGennaroSchoenmakers97}
Ronald Cramer, Rosario Gennaro, and Berry Schoenmakers.
\newblock A Secure and Optimally Efficient Multi-Authority Election Scheme.
\newblock In \emph{Advances in Cryptology -- {EUROCRYPT} '97}, volume 1233 of
\emph{Lecture Notes in Computer Science}, pages 103--118. Springer, 1997.

\bibitem{DaianGoldfederKellLiZhaoBentovBreidenbachJuels19}
Philip Daian, Steven Goldfeder, Tyler Kell, Yunqi Li, Xueyuan Zhao, Iddo
Bentov, Lorenz Breidenbach, and Ari Juels.
\newblock Flash Boys 2.0: Frontrunning in Decentralized Exchanges, Miner
Extractable Value, and Consensus Instability.
\newblock In \emph{2020 IEEE Symposium on Security and Privacy}, pages 910--927,
2020.

\bibitem{DornerKonstantinovPashalievVechev23}
Florian~E. Dorner, Nikola Konstantinov, Georgi Pashaliev, and Martin Vechev.
\newblock Incentivizing Honesty among Competitors in Collaborative Learning and
Optimization.
\newblock In \emph{Advances in Neural Information Processing Systems 36}, 2023.

\bibitem{DziembowskiFaustLuhn24}
Stefan Dziembowski, Sebastian Faust, and Jannik Luhn.
\newblock Shutter Network: Private Transactions from Threshold Cryptography.
\newblock \emph{Cryptology ePrint Archive}, Paper 2024/1981, 2024.

\bibitem{FranklinReiter96}
Matthew~K. Franklin and Michael~K. Reiter.
\newblock The Design and Implementation of a Secure Auction Service.
\newblock \emph{IEEE Transactions on Software Engineering}, 22(5):302--312,
1996.

\bibitem{HalpernTeague04}
Joseph Halpern and Vanessa Teague.
\newblock Rational Secret Sharing and Multiparty Computation: Extended
Abstract.
\newblock In \emph{Proceedings of the Thirty-Sixth Annual ACM Symposium on
Theory of Computing}, 2004.

\bibitem{KavousiLeJovanovicDanezis23}
Alireza Kavousi, Duc~V. Le, Philipp Jovanovic, and George Danezis.
\newblock BlindPerm: Efficient {MEV} Mitigation with an Encrypted Mempool and
Permutation.
\newblock In \emph{29th International Conference on Principles of Distributed
Systems (OPODIS 2025)}, volume 361 of \emph{Leibniz International Proceedings in
Informatics (LIPIcs)}, pages 36:1--36:21, 2026.

\bibitem{KulkarniDiamandisChitra22}
Kshitij Kulkarni, Theo Diamandis, and Tarun Chitra.
\newblock Towards a Theory of Maximal Extractable Value {I}: Constant Function
Market Makers.
\newblock \emph{arXiv preprint arXiv:2207.11835}, 2022.

\bibitem{LiHuWangDuanWang23}
Rujia Li, Xuanwei Hu, Qin Wang, Sisi Duan, and Qi Wang.
\newblock Transaction Fairness in Blockchains, Revisited.
\newblock \emph{IEEE Transactions on Dependable and Secure Computing},
23(1):752--765, 2026.

\bibitem{MailathSamuelson06}
George~J. Mailath and Larry Samuelson.
\newblock \emph{Repeated Games and Reputations: Long-Run Relationships}.
\newblock Oxford University Press, 2006.

\bibitem{MaterwalaNaikTahaAbedSvetinovic24}
Huned Materwala, Shraddha~M. Naik, Aya Taha, Tala Abdulrahman Abed, and Davor
Svetinovic.
\newblock Maximal Extractable Value in Decentralized Finance: Taxonomy,
Detection, and Mitigation.
\newblock \emph{IEEE Transactions on Services Computing}, 18(6):4386--4407, 2025.

\bibitem{MiltersenNielsenTriandopoulos09}
Peter Bro Miltersen, Jesper Buus Nielsen, and Nikos Triandopoulos.
\newblock Privacy-Enhancing Auctions Using Rational Cryptography.
\newblock In \emph{Advances in Cryptology -- {CRYPTO} 2009}, volume 5677 of
\emph{Lecture Notes in Computer Science}, pages 541--558. Springer, 2009.

\bibitem{MouchetTroncosoPastorizaBossuatHubaux21}
Christian Mouchet, Juan Troncoso-Pastoriza, Jean-Philippe Bossuat, and
Jean-Pierre Hubaux.
\newblock Multiparty Homomorphic Encryption from Ring-Learning-with-Errors.
\newblock \emph{Proceedings on Privacy Enhancing Technologies},
2021(4):291--311, 2021.

\bibitem{NasrulinIshmaevDecouchantPouwelse23}
Bulat Nasrulin, Georgy Ishmaev, J\'er\'emie Decouchant, and Johan Pouwelse.
\newblock L{\O}: An Accountable Mempool for {MEV} Resistance.
\newblock In \emph{Proceedings of the 24th ACM/IFIP International Middleware
Conference}, pages 98--110, 2023.

\bibitem{RaisaroTroncosoPastorizaMisbachSousaPradervandMissiagliaMichielinFordHubaux19}
Jean~Louis Raisaro, Juan~Ram\'on Troncoso-Pastoriza, Micka\"el Misbach,
Jo\~ao~S\'a Sousa, Sylvain Pradervand, Edoardo Missiaglia, Olivier Michielin,
Bryan Ford, and Jean-Pierre Hubaux.
\newblock {MedCo}: Enabling Secure and Privacy-Preserving Exploration of
Distributed Clinical and Genomic Data.
\newblock \emph{IEEE/ACM Transactions on Computational Biology and
Bioinformatics}, 16(4):1328--1341, 2019.

\bibitem{RondeletKilbourn23}
Antoine Rondelet and Quintus Kilbourn.
\newblock Mempool Privacy: An Economic Perspective.
\newblock \emph{arXiv preprint arXiv:2307.10878}, 2023.

\bibitem{SavPyrgelisTroncosoPastorizaFroelicherBossuatSousaHubaux21}
Sinem Sav, Apostolos Pyrgelis, Juan~Ram\'on Troncoso-Pastoriza, David
Froelicher, Jean-Philippe Bossuat, Jo\~ao~S\'a Sousa, and Jean-Pierre Hubaux.
\newblock {POSEIDON}: Privacy-Preserving Federated Neural Network Learning.
\newblock In \emph{Proceedings of the Network and Distributed System Security
Symposium ({NDSS})}, 2021.

\bibitem{Shapley62}
Lloyd~S. Shapley.
\newblock Simple Games: An Outline of the Descriptive Theory.
\newblock \emph{Behavioral Science}, 7(1):59--66, 1962.

\bibitem{StengeleRaiberMuellerQuadeHartenstein21}
Oliver Stengele, Markus Raiber, J\"orn M\"uller-Quade, and Hannes Hartenstein.
\newblock {ETHTID}: Deployable Threshold Information Disclosure on Ethereum.
\newblock In \emph{Proceedings of the 3rd International Conference on
Blockchain Computing and Applications}, pages 127--134, 2021.

\bibitem{TaylorZwicker99}
Alan~D. Taylor and William~S. Zwicker.
\newblock \emph{Simple Games: Desirability Relations, Trading,
Pseudoweightings}.
\newblock Princeton University Press, 1999.

\bibitem{TorresCaminoState21}
Christof~Ferreira Torres, Ramiro Camino, and Radu State.
\newblock Frontrunner Jones and the Raiders of the Dark Forest: An Empirical
Study of Frontrunning on the {E}thereum Blockchain.
\newblock In \emph{30th {USENIX} Security Symposium}, pages 1343--1359, 2021.

\bibitem{WadhwaMaThieryMonnotZanoliniZhangNayak25}
Sarisht Wadhwa, Julian Ma, Thomas Thiery, Barnabe Monnot, Luca Zanolini, Fan
Zhang, and Kartik Nayak.
\newblock {AUCIL}: An Inclusion List Design for Rational Parties.
\newblock \emph{Cryptology ePrint Archive}, Paper 2025/194, 2025.

\end{thebibliography}
\end{document}